\documentclass[12pt]{article}
\usepackage{amsmath}
\usepackage{times}
\usepackage{graphicx}
\usepackage{color}
\usepackage{multirow}
\usepackage{amsfonts}
\usepackage{amssymb}
\usepackage[authoryear]{natbib}
\usepackage{rotating}
\usepackage{bbm}
\usepackage{latexsym}

\newtheorem{theorem}{Theorem}[section]

\newtheorem{proposition}[theorem]{Proposition}

\newenvironment{proof}[1][Proof]{\begin{trivlist}
\item[\hskip \labelsep {\bfseries #1}]}{\end{trivlist}}

\newenvironment{remark}[1][Remark]{\begin{trivlist}
\item[\hskip \labelsep {\bfseries #1}]}{\end{trivlist}}

\newcommand{\qed}{\nobreak \ifvmode \relax \else
      \ifdim\lastskip<1.5em \hskip-\lastskip
      \hskip1.5em plus0em minus0.5em \fi \nobreak
      \vrule height0.75em width0.5em depth0.25em\fi}

\newcommand{\captionfonts}{\normalsize}

\makeatletter  
\long\def\@makecaption#1#2{%
  \vskip\abovecaptionskip
  \sbox\@tempboxa{{\captionfonts #1: #2}}%
  \ifdim \wd\@tempboxa >\hsize
    {\captionfonts #1: #2\par}
  \else
    \hbox to\hsize{\hfil\box\@tempboxa\hfil}%
  \fi
  \vskip\belowcaptionskip}
\makeatother   

\begin{document}
\hspace{13.9cm}1

\ \vspace{20mm}\\

{\LARGE The geometry of Tempotronlike problems}

\ \\
{\bf \large Konrad P. Kording$^{\displaystyle 1, \displaystyle 2}$}\\
{$^{\displaystyle 1}$Rehabilitation Institute of Chicago}\\
{$^{\displaystyle 2}$Northwestern University, 345 E Superior Street, Chicago, IL 60611}\\

{\bf Keywords:} Tempotron, Computational geometry, Polytope coverage problem

\thispagestyle{empty}
\markboth{}{NC instructions}
\ \vspace{-0mm}\\
%
\begin{center} {\bf Abstract} \end{center}
In the discrete Tempotron learning problem a neuron receives time varying inputs and for a set of such input sequences ($\mathcal S_-$ set) the neuron must be sub-threshold for all times while for some other sequences ($\mathcal S_+$ set) the neuron must be super threshold for at least one time.  Here we present a graphical treatment of a slight reformulation of the tempotron problem. We show that the problem's general form is equivalent to the question if a polytope, specified by a set of inequalities, is contained in the union of a set of equally defined polytopes. Using recent results from computational geometry, we show that the problem is W[1]-hard. This phrasing gives some new insights into the nature of gradient based learning algorithms. A sampling based approach can, under certain circumstances provide an approximation in polynomial time. Other problems, related to hierarchical neural networks may share some topological structure.


\section{Introduction}
The tempotron (G\"utig and Sompolinsky 2006) problem occurs when a given neuron is supposed to learn to respond to time-varying signals. Let $\bf I(t)\in \mathbb{R}^N$ be a time varying input vector where $t\in \mathbb{N}^+$ runs from $1$ to $T$. 
The neurons output $V(t)\in \Re$ depends on its inputs according to:
\begin{equation}
V(t)={\bf W}^{\intercal}\sum_{\Delta t}{\bf I}(t-\Delta t)K(t-\Delta t)+V_{rest}
\end{equation}
Here $\bf W \in \mathbb{R}^N$ is a set of weights characterizing the neuron's input/output relationship. 
$K(t-\Delta t)$ is the Kernel characterizing the relation of the current potential on the history of past inputs. Keep in mind that this kernel contains both the temporal properties of the neuron $\Delta t \in \mathbb{N}$ runs from $0$ to $t-1$, over the time passed since the start of the stimulation. 
 
In the specific example in the original paper the chosen Kernel had the form$K(t-t_i)=V_0(\exp[-(t-t_i)/\tau]-\exp[-(t-t_i)/\tau_s])$ with free parameters $\tau$ and $\tau_s$ that characterize the decay time constants of membrane integration and synaptic currents, respectively. All inputs $\bf I$ were constrained to be binary.

The neuron is considered to "spike" in response to a time varying input stimulus $\exists t \in \mathbb{N}^+$ for which $V(t)>V_{thresh}$. There are two sets relevant for learning, each consisting of input matrices  ${\bf I}:=({\bf I}(1),{\bf I}(2),\cdots, {\bf I}(T)) \in \mathbb{R}^{N \times T}$. There is the set $\mathcal S_{-}$, of $J$ time varying inputs for which there must be no spike. There also is the set $\mathcal S_{+}$ of $K$ inputs for which there must be at least one spike. The Tempotron learning problem asks for a $\bf W$ to be given that satisfies the constraints given by the sets $\mathcal S_{+}$ and $\mathcal S_{-}$. The corresponding decision problem asks if a solution exists.

We now define the reformulated input ${\bf J}(t)$:
\begin{equation}
{\bf J}(t)= \sum_{\Delta t}{\bf I}(t-\Delta t)K(t-\Delta t)
\end{equation}
and obtain as a simplified criterion for spiking:

\begin{equation}
{\bf W}^{\intercal}{\bf J}(t)>V_{thresh}-V_{rest}
\label{linearConstraint}
\end{equation}
\section{Prior Work}
In a broad set of analyses using simulations and mathematical analysis the Tempotron has been analyzed. For example, it was asked  (Ram, Monasson and Sompolinski, 2010) what the capacity of the Tempotron was.

\section{Reformulating the problem}
In this paper we will now slightly rewrite the original problem to make it amenable to a graphical treatment. First, we will assume that for all $i \in (1,2,\cdots, N)$ the individual weight  $W_i$ is upper and lower bounded by $W_{min} \in \Re$ and $W_{max} \in \Re$, respectively. This is primarily to allow us to work with bounded as opposed to unbounded polytopes and to give them a defined volume. Second we will assume that the input vectors $\bf J$ can not only reach all values compatible with spiking inputs $\bf I$ and a given set of Kernels but that arbitrary such inputs are possible. This seems like a rather good approximation to the real problem and assumes that continuous inputs $\bf J$ can be approximated by the right choice spikes along with the right Kernels. Lastly, we assume without loss of generality that $V_{thresh}-V_{rest}=1$. We call this continuous version of the Tempotron {\it C-Tempotron}.

\begin{remark}
It seems very likely to the author that the Tempotron and the C-Tempotron are actually problems of identical hardness but this would have to be the topic of a more careful analysis.
\end{remark}

\section{Geometrical version of the tempotron problem}
 Here we will consider the geometry of space of weights $\mathcal X \subset \mathbb{R}^N$ that are compatible with $\mathcal S_{-}$ and call that inclusion set $\mathcal X_{inc}\subset \mathbb{R}^N$. We will also consider the space of all possible $\mathcal X \subset \mathbb{R}^N$ that is {\it incompatible} with $\mathcal S_{+}$ and call that exclusion set $\mathcal X_{exc}\subset \mathbb{R}^N$. We formalize the necessary and sufficient conditions for the existence of a solution to the Tempotron problem in the following proposition.
 
\begin{proposition}
\emph{The C-Tempotron problem is equivalent to the problem of answering if a polytope in $\mathbb{R}^N$ is in its entirety covered by a union of other polytopes in $\mathbb{R}^N$}.
\end{proposition}

However to prove this proposition we first need to consider the nature of both inclusion and exclusion sets.
 \begin{figure}
\begin{center}
\includegraphics[width=0.5\textwidth]{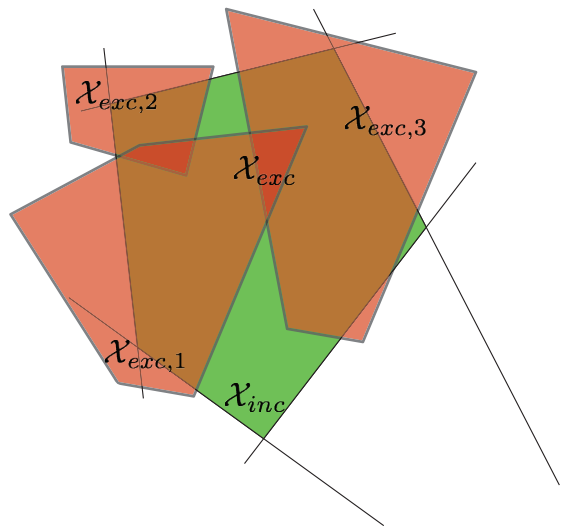}
\end{center}
\caption{Geometrical interpretation: Here we see an example of the geometrical view on the problem. Green is the inclusion set. Red is the exclusion set. This exclusion set is the union of its constituent exclusion sets. The solution set is the part of the inclusion set not covered by the exclusion set, the part displayed as green}
 \label{ModelKinds}
\end{figure}

\subsection{The inclusion set} 
We will first discuss the nature of the inclusion set $\mathcal X_{inc}$ that is defined by the trials in $\mathcal S_{-}$.
A trial that is in the $\mathcal S_{-}$ set implies that for all times $t$ from $1$ to $T$ the inequality ${\bf W}^{\intercal}{\bf I}_{t}<1$ (no spike) is obeyed. Then we obtain $\mathcal X_{inc}$ by considering the inequalities for each stimulus $j$ and each time point $t$:
${\bf W}^{\intercal}{\bf I}_{j,t}<1$. 
As this set is defined by $J \times T$ linear inequalities, along with the $W_{min}<W_{i}<W_{max}$ inequalities the set $\mathcal X_{inc} \subseteq \mathbb{R}^N$ defines a convex polytope.

\subsection{The exclusion set}
Now we will discuss the nature of the exclusion set $\mathcal X_{exc}$ that is defined by the trials in $\mathcal S_{+}$.
A trial that is in the $\mathcal S_{+}$ set implies that $\exists t$ so that ${\bf W}^{\intercal}I_{t}\geq 1$ (spike) is obeyed. The solution set for a given $j$ now is the union of subspaces defined by ${\bf W}^{\intercal}I_{jt}\geq 1$. The inverse of this, the set of $\bf W$ that are $\bf not$ solutions ($\mathcal X_{exc, j}$) now is a polytope again defined by the $T$ inequalities:
 ${\bf W}^{\intercal}I_{j,t}<1$.
 For each stimulus $j \in \mathcal S_{+}$ these $K+2N$ inequalities, along with the bounds on $\bf W$ thus define a polytope  $\mathcal X_{exc,j} \subseteq \mathbb{R}^N$.
The set of all solutions excluded by $\mathcal S_{+}$ now is the union of the polytopes defined by each stimulus:
\begin{equation}
\mathcal X_{exc}=\bigcup_{j=1}^{J} \mathcal X_{exc,j }
\end{equation}
It is important to note here that the exclusion set, being a union of polytopes, will generally not be convex.

\subsection{The solution set}

A solution to the C-Tempotron problem exists {\it iff} the inclusion set is not in its entirely covered by the exclusion sets:
 \begin{equation}
 \mathcal X_{inc} \subseteq \bigcup_{k=1}^K \mathcal X_{exc,j }=\mathcal X_{exc}
 \end{equation}
 and each of the $\mathcal X$ are polytopes.
 As this problem is exactly one way of specifying the polytope coverage problem we have proven the proposition.

\subsection{True equivalence}
Per construction, each tempotron problem has an exactly analogous polytope coverage problem. But is it the case that each polytope coverage problem has an equivalent C-tempotron problem? The reader will probably have observed that the definitions are exactly equivalent apart from the fact that in the standard polytope definition each inequality has a potentially distinct bias. But it turns out that this bias can be built into the Tempotron through a slightly altered construction principle. We will first augment each input vector $\bf I$ with the desired bias (say at its end). We will next set the corresponding weight to be 1 using the following procedure. We construct one stimulus that we append to $\mathcal S_{+}$ and one stimulus to $\mathcal S_{-}$ which are both vectors only along the dimension $N+1$ and are $\epsilon$ away from one another, effectively clamping $W_{N+1}$ to 1. We now have a construction to assign an equivalent Tempotron problem to each polytope coverage problem and vice versa.
 
\section{Implications from the geometric framework}
\begin{proposition} 
\emph{The C-Tempotron problem is at least W[1] hard}
\end{proposition}
\begin{proof}
We can construct a C-Tempotron such that each of the polytopes is parallel to the axes. For this scenario Chan(2008) in his Observation 5.2 has shown the problem to be NP hard, or more specifically W[1] hard. There are many other famous problems of complexity W[1] including problems about graphs and their cliques and the acceptance of certain Turing machines.
\end{proof}

\subsection{Insights into the gradient based method}
The learning algorithm introduced in the original Tempotron problem implements the following graphical operation. To satisfy the constraints of $\mathcal S_{-}$ we take all violated constraints and try to move away from them. To satisfy the constraints of $\mathcal S_{+}$ for each stimulus we take the one constraint that is closest to being obeyed and move into that direction. This can easily make a solution impossible to arrive to if there are extra forbidden areas in between the algorithms start position and the actual solution space. For example, if $N=1$ and the starting position is close to $W_{max}$ and we have a single excluded polytope that ranges from $W_{max}/2$ to $W_{max}$. Then the closest inequality to the starting point is always going to be $W_{max}$ despite there being no solution set in that directions and despite half of the available space actually being solutions.

\subsection{Sampling based approximate algorithms}
We can sample from a polytope in polynomial time (Kannan, Lov‡sz, Simonovits, 1997).  So one simple algorithm is for us to sample from  $\mathcal X_{inc}$. For each sample we can then in $O(T\times K \times N)$ determine if the sample is in $\mathcal X_{exc}$. We can thus produce each such answer in polynomial time. 

Let us now define the remaining volume of the solution space as the Lebesques volume $\lambda (\mathcal X_{inc} \backslash \mathcal X_{exc})$. Then if a sample exists we should each iteration find it with probability:
\begin{equation}
p=\frac{\lambda (\mathcal X_{inc} \backslash \mathcal X_{exc})}{\lambda (\mathcal X_{inc})}
\end{equation}
We are sampling randomly from the set of weights that are compatible with the $\mathcal S_{-}$ set and hope that they randomly obey the $\mathcal S_{+}$ set. Our procedure thus either produces a solution $\bf W$ to the problem, or it gives us a progressively lower estimate about the probability of the existence of a solution. If there exists a solution then our probability of finding it is $(1-p)^n$. 

It is thus clear that we can not obtain certainty. However, given an assumption of $p$ and an assumption of $\delta$, a false negative rate we are willing to tolerate, there will be a number of samples $n$ that allows us to solve the problem. 
By solving the inequality $(1-p)^n\leq \delta$.  We obtain that this number of samples should be:
\begin{equation}
n \geq \log (1/\delta)/p
\end{equation} 

So if we are willing to accept a fixed proportion of false positives then we can find the solution in $O(N^5\log (1/\delta)/p)$ steps.  However, this term will under some circumstances still be exponential with the dimensionality of the problem. For example, if each exclusion set $\mathcal X_{exc,j }$ covers a random fixed proportion of the remaining space then $p$ decreases exponentially with the dimensionality of the problem. It all depends on how Tempotron problems are actually constructed (or given by nature). 

\section{Discussion}
Here we have shown that the C-Tempotron problem and the polytope coverage problem are equivalent. We have shown that they are at least of W[1] complexity. Lastly, we have also shown how there may be approximations that are more efficient under certain assumptions.

Our problem is formulated in discrete time (see Rubin, Monasson, Sompolinski 2010) but there is a lot of interest in continuous time approaches. It is not entirely clear what would change in the continuous case. Representing the problem would become far harder. And simple approximation strategies, e.g. inner and outer bounding by finite polytopes, may well lead to fundamental changes in high dimensions. Understanding finale time Tempotrons sounds like an exciting but challenging endeavor. 

Our approach may shed some light on empirical findings about the tempotron. It was observed that under certain circumstances the solution space is fractured into many small domains, spread roughly evenly across the weight space. This could naturally relate to the fact that ex ante, there is no reason for the polytopes to cluster in any one region of space. This in term may explain why the proposed geometrical algorithm appears to converge quite well (e.g. because each starting point is close to a potential solution). 

The Tempotron bears some computational similarity with other problems occurring in the field of computation. If we have a system of binary elements with linear decision boundaries and an output function that is the {\it or} function then the problem is equivalent to the tempotron problem. Arguably, problems that are similar happen in many cases in the context of neural networks with max pooling.
 
\subsection*{Acknowledgements}
I am grateful to Sara Solla for making me write this up. To Robert Guetig for discussions. To Hans Raj Tiwary for helping me understand why the problem is hard. To Ran Rubin for helping me understand the Tempotron more deeply. And to Mohammad Azar for impressive help and discussions.

\section{References}
Chan, Timothy M. "A (slightly) faster algorithm for Klee's measure problem." Proceedings of the twenty-fourth annual symposium on Computational geometry. {\it ACM}, 2008.\newline

G\"utig, Robert, and Haim Sompolinsky. "The tempotron: a neuron that learns spike timingÐbased decisions." {\it Nature neuroscience} 9.3 (2006): 420-428.\newline

Kannan, Ravi, L\'aszl\'o Lov\'asz, and Mikl\'os Simonovits. "Random walks and an O*(n5) volume algorithm for convex bodies." {\it Random structures and algorithms} 11.1 (1997): 1-50.\newline

Rubin, Ran, R\'emi Monasson, and Haim Sompolinsky. "Theory of spike timing-based neural classifiers." {\it Physical review letters} 105.21 (2010): 218102.\newline

\end{document}